\documentclass[12pt]{article}
\usepackage[a4paper,left=15mm,top=20mm,right=15mm,bottom=20mm]{geometry}

\usepackage{graphicx}
\usepackage{algorithm}
\usepackage[noend]{algpseudocode}
\usepackage{amsmath, amsfonts, amssymb, amsthm}
\usepackage{mathtools}
\usepackage[colorlinks]{hyperref}
\hypersetup{
	citecolor=blue,
   linkcolor=blue,
}

\usepackage{authblk}
\usepackage[numbers,sort&compress]{natbib}
\usepackage[font=footnotesize,labelfont=bf]{caption}

\algrenewcommand\algorithmicrequire{\textbf{Input:}}
\algrenewcommand\algorithmicensure{\textbf{Output:}}

\newtheorem{theorem}{Theorem}[section]

\newtheorem{lemma}[theorem]{Lemma}%[section]
\newtheorem{definition}[theorem]{Definition}%[section]
\newtheorem{corollary}[theorem]{Corollary}%[section]
\newcommand{\eps}{\ensuremath{\varepsilon}}
\newcommand{\Eps}{\ensuremath{\mathcal{E}}}
\newcommand{\Ts}{\ensuremath{T^*}}
\newcommand{\ls}{\ensuremath{\lambda^*}}
\newcommand{\irr}[1]{\ensuremath{[#1]_{\textrm{irr}}}}
\newcommand{\lca}{\ensuremath{\operatorname{lca}}}
\newcommand{\Mo}{\ensuremath{M^{\otimes}}}
\DeclareMathOperator{\Aho}{Aho}
\newcommand{\cl}{\ensuremath{\operatorname{cl}}}
\newcommand{\IT}[1]{{R_{#1}}}

\providecommand{\keywords}[1]{\textbf{\textit{Keywords: }} #1}

\makeindex             % used for the subject index
\begin{document}

\title{Generalized Fitch Graphs: Edge-labeled Graphs that are explained by Edge-labeled Trees}

\author[]{Marc Hellmuth}

\affil[]{\footnotesize Dpt.\ of Mathematics and Computer Science, University of Greifswald, Walther-
  Rathenau-Strasse 47, D-17487 Greifswald, Germany\\ \smallskip

	Saarland University, Center for Bioinformatics, Building E 2.1, P.O.\ Box 151150, D-66041 Saarbr{\"u}cken, Germany\\ \smallskip
	Email: \texttt{mhellmuth@mailbox.org}}

\date{\ }

\maketitle

\abstract{ 
Fitch graphs $G=(X,E)$ are di-graphs that are explained by $\{\otimes,1\}$-edge-labeled rooted trees with leaf set $X$: there is an arc $xy\in E$ if and only if the unique path in $T$ that connects the least common ancestor $\lca(x,y)$ of $x$ and $y$ with $y$ contains at least one edge with label ``$1$''. In practice, Fitch graphs represent xenology relations, i.e., pairs of genes $x$ and $y$ for which a horizontal gene transfer happened along the path from $\lca(x,y)$ to $y$.  

In this contribution, we generalize the concept of Fitch graphs and  consider complete di-graphs $K_{|X|}$  with vertex set $X$ and a map $\eps$ that assigns to each arc $xy$ a unique label  $\eps(x,y)\in M\cup \{\otimes\}$, where $M$ denotes an arbitrary set of symbols.  A di-graph $(K_{|X|},\eps)$ is a generalized Fitch graph if there is an $M\cup \{\otimes\}$-edge-labeled tree $(T,\lambda)$ that can explain $(K_{|X|},\eps)$. 

We provide a simple characterization of generalized Fitch graphs $(K_{|X|},\eps)$ and give an $O(|X|^2)$-time algorithm for their recognition as well as for the reconstruction of the unique least-resolved phylogenetic tree that explains $(K_{|X|},\eps)$. 
}
	
\smallskip
\noindent
\sloppy

\sloppy
\section{Introduction}
Edge-labeled graphs that can be explained by \emph{vertex-labeled} trees
have been widely studied and range from cographs \cite{Corneil:81,HHH+13} and di-cographs \cite{Crespelle:06} 
to so-called unp-2-structures \cite{ER1:90,ER2:90,EHPR:96}, 
symbolic ultrametrics \cite{Boecker:98,Hellmuth:16a} or 
three-way symbolic tree-maps \cite{HSM:18,GLW:18}. 
Besides their structural attractiveness, those types of graphs play an important
role in phylogenomics, i.e., the reconstruction of the evolutionary history
of genes and species.
By way of example, the concept of \emph{orthologs}, that is, pairs of genes from different species that arose
from a speciation event \cite{Fitch:70}, is
of fundamental importance in many fields of mathematical and computational
biology, including the reconstruction of evolutionary relationships across
species \cite{DBH-survey05,Hellmuth:15a} or  functional genomics and gene
organization in species \cite{GK13,TG+00}.
The orthology relation $\Theta$ is explained by vertex-labeled trees, i.e., 
a gene pair $(x,y)$ is contained in $\Theta$ if and only if the least common ancestor
of $x$ and $y$ is labeled as a speciation event. 
The graph representation of $\Theta$ must necessarily be a co-graph  \cite{HHH+13,Boecker:98} 
and provides direct information on the gene history 
as well as on the history of the species \cite{HW:16book,Hellmuth:15a}

In contrast, \emph{xenology} as defined by  Walter M. Fitch \cite{Fitch:00}
is explained by \emph{edge-labeled} rooted phylogenetic trees: a gene $y$ is xenologous with
respect to $x$,  if and only if the unique path from the least common ancestor $\lca(x,y)$
to $y$ in the gene tree contains a transfer edge. 
In other words, the xenology relation is explained by an $\{\otimes,1\}$-edge-labeled 
rooted tree, where an edge with label ``$1$'' is a transfer edge and an edge with 
label ``$\otimes$'' is a non-transfer edge. It has been shown  by Gei{\ss} et al.\ \cite{Geiss:17}
that the xenology relation forms a \emph{Fitch graph}, that is, an 
$\{\otimes,1\}$-edge-labeled di-graph which is characterized by the absence
of eight forbidden subgraphs on three vertices. 
Moreover, for a given  Fitch graph $\mathcal{F}$ it is possible to reconstruct the
unique minimally resolved 
phylogenetic tree that explains $\mathcal{F}$ in linear time.

A further example of graphs and relations that are defined in terms of edge-labeled trees are 
the single-1-relations $\mathrel{\overset{1}{\sim}}$
and $\mathrel{\overset{1}{\rightharpoonup}}$  \cite{Hellmuth:17a}. 
These relations are defined by 
the existence of a \emph{single} edge with label ``$1$'' along the
connecting path of two genes and capture the structure of so-called 
rare genomic changes (RGCs). 
RGCs  
	have been proven to be phylogenetically informative and helped to resolve many of the phylogenetic questions where
sequence data lead to conflicting or equivocal results, see e.g.\
\cite{Boore:06,Deeds:05,Donath:14a,Dutilh:08,Krauss:08a,Prohaska:04a,Rogozin:05,Rokas:00,Shedlock:00}.
%In \cite{Hellmuth:17a},  the structure of single-1-relations was characterized. 
%It was shown that the graph representation 
%of a single-1-relation is always a forest and if the single-1-relation is connected, there
%is a unique minimally resolved tree that explains the relation. The same holds true
%for the connected components of an arbitrary single-1-relation 

In summary, edge-labeled graphs (or equivalently, binary relations) 
that can be explained by \emph{edge-labeled} trees 
provide important information about the evolutionary history of 
the underlying genes. However, for such type of graphs 
only few results are available \cite{Geiss:17,Hel:18,Hellmuth:17a}.

In this contribution, we extend the notion of xenology and Fitch graphs to 
\emph{generalized} Fitch graphs, that is, 
di-graphs that can be derived from  $\{\otimes, 1,\dots,m\}$-edge
labeled trees, or equivalently, edge-labeled di-graphs 
that can be explained by such trees.  
We show that these graphs are characterized by four simple conditions
that are defined in terms of edge-disjoint subgraphs. Moreover, 
we give an $O(|X|^2)$-time recognition algorithm for generalized 
Fitch graphs on a set of vertices $X$ and the reconstruction of the unique least-resolved 
phylogenetic tree that explains them.

\section{Preliminaries}

\subsection{Trees,  Di-Graphs and Sets}

A \emph{rooted tree} $T=(V,E)$ (on $X$) is 
an acyclic connected graphs with leaf set $X$, set of {\em inner} vertices $V^0=V\setminus X$ and
one distinguished inner vertex $\rho_T\in V^0$ that is called the \emph{root of T}.
In what follows, we consider always \emph{phylogenetic} trees $T$, that is,
rooted trees such that the root $\rho_T$ has at least
degree $2$ and every other inner vertex $v\in V^0\setminus\{\rho_T\}$ has
at least  degree $3$.

We call $u\in V$ an \emph{ancestor} of $v\in V$, $u\succeq_T v$, and $v$ a
\emph{descendant} of $u$, $v\preceq_T u$, if $u$ lies on the unique path
from $\rho_T$ to $v$. We write $v\prec_T u$ ($u\succ_T v$) for $v\preceq_T
u$ ($u\succeq_T v$) and $u\neq v$. 
Edges that are incident to a leaf are called \emph{outer
edges}. Conversely, \emph{inner edges} do only contain inner vertices.
For a non-empty subset $Y\subseteq X$ of leaves, \emph{the least common
  ancestor of} $Y$, denoted as $\lca_T(Y)$, is the unique
$\preceq_T$-minimal vertex of $T$ that is an ancestor of every vertex in
$Y$. We will make use of the simplified notation
$\lca_T(x,y):=\lca_T(\{x,y\})$ for $Y=\{x,y\}$ and we will omit the
explicit reference to $T$ whenever it is clear which tree is
considered. 
For a subset $Y\subseteq X$ of leaves, 
the tree $T(Y)$ with root $\lca_T(Y)$ has leaf set $Y$ and consists of all paths in $T$
that connect the leaves in $Y$.  
The \emph{restriction} $T_{|Y}$ of $T$ to some subset $Y\subseteq X$ 
 is the rooted tree obtained from $T(Y)$ by suppressing all vertices of
degree $2$ with the exception of the root $\rho_T$ if
$\rho_T\in V(T(Y))$.

A \emph{contraction} of an edge $e = xy$ in a tree $T$ refers to the removal of $e$ and
identification of $x$ and $y$.
We say that a rooted tree $T$ on $L$ \emph{displays} a root tree $T'$ on $L'$, in symbols $T'\le T$, if $T'$ can
be obtained from $T(L')$ by a sequence of edge contractions.  
%We
%write $T'<T$ if $T'\le T$ and $T'\ne T$. 
If $T'\le T$, then we also say that $T$ \emph{refines} $T'$. 
%It is easy to see that a phylogenetic tree $T$ refines a rooted tree
%$T'$ if and only if $\mathcal{C}(T')\subseteq \mathcal{C}(T)$.
% Moreover, $T|L'\le T$, i.e., $T$ displays the
%restrictions $T|L'$ to all subsets $L'\subseteq L$. 

\emph{Rooted triples} are binary rooted phylogenetic trees on three leaves.
We write $ab|c$ for the rooted triple with leaves $a,b$ and $c$, if the
path from its root to $c$ does not intersect the path from $a$ to $b$.  The
definition of ``display'' implies that a triple $ab|c$ with $a,b,c\in L$ is
\emph{displayed} by a rooted tree $T$ if $\lca(a,b)\prec_T \lca(a,b,c)$.
The set of all triples that are displayed by $T$ is denoted by $r(T)$.
A set of rooted triples $R$ is called \emph{consistent} if there exists a phylogenetic tree $T$ on
$L_R\coloneqq\bigcup_{ab|c\in R} \{a,b,c\}$ that displays $R$, i.e.,
$R\subseteq r(T)$.  As shown in
\cite{aho_inferring_1981} there is a polynomial-time algorithm, usually
referred to as \texttt{BUILD}
\cite{semple_phylogenetics_2003,steel_phylogeny:_2016}, that takes a set
$R$ of triples as input and either returns a particular phylogenetic tree
$\Aho(R)$ that displays $R$, or recognizes $R$ as inconsistent.

%Rooted triples are widely used in the context of supertree reconstruction
%because every phylogenetic tree $T$ is \emph{identified} by its triple set $r(T)$
%\cite{semple_phylogenetics_2003}.  To be more precise, a 
A set of rooted triples $R$ \emph{identifies} a tree $T$ with leaf set
$L_R$ if $R$ is displayed by $T$ and every other tree $T'$ that displays
$R$ is a refinement of $T$. A rooted triple $ab|c\in r(T)$
\emph{distinguishes} an edge $uv$ in $T$ iff $a$, $b$, and $c$ are
descendants of $u$; $v$ is an ancestor of $a$ and $b$ but not of $c$; and
there is no descendant $v'$ of $v$ for which $a$ and $b$ are both
descendants. In other words, $ab|c\in r(T)$ distinguishes the edge $uv$
if $\lca(a,b)=v$ and $\lca(a,b,c)=u$.

The requirement that a set $R$ of triples is consistent, and thus, that
there is a tree displaying all triples, makes it possible to infer new
triples from the trees that display $R$ and to define a \emph{closure
  operation} for $R$
\cite{grunewald_closure_2007,bryant_extension_1995,HS:17,Bryant97}.  Let
$\langle R \rangle$ be the set of all rooted trees with leaf set $L_R$ that
display $R$.  The closure of a consistent set of rooted triples $R$ is
defined as
\begin{equation*}
  \cl(R) = \bigcap_{T\in \langle R\rangle} r(T).
\end{equation*}
Hence, a triple $r$ is contained in the closure $\cl(R)$ if all trees that
display $R$ also display $r$.  This operation satisfies the usual three
properties of a closure operator \cite{bryant_extension_1995}, namely: (i)
expansiveness, $R \subseteq \cl(R)$; (ii) isotony, $R' \subseteq R$ implies
that $\cl(R')\subseteq \cl(R)$; and (iii) idempotency,
$\cl(\cl(R))=\cl(R)$. Since $T \in \langle r(T)\rangle$, it is easy to see
that $\cl(r(T))=r(T)$ and thus, $r(T)$ is always closed.

For later reference , we give here an important result 
from \cite{grunewald_closure_2007} that is
closely related to the \texttt{BUILD} algorithm.
\begin{lemma}\label{g1}
  Let $T$ be a phylogenetic tree and let $R$ be a set of rooted triples.
  Then, $R$ identifies $T$ if and only if $\cl(R)=r(T)$.  Moreover, if $R$
  identifies $T$, then $\Aho(R)=T$.
\end{lemma}

In this contribution, we will
consider phylogenetic trees $T=(V,E)$ together with an edge-labeling map 
$\lambda:E\to M\cup\{\otimes\}$,
where $M=\{1,\dots |M|\}$ denotes a non-empty set of symbols and  
we write $(T,\lambda)$. Edges that have label $m\in M\cup \{\otimes\}$ 
are called \emph{$m$-edges}. Furthermore, $\Mo$ will always denote 
the set $M\cup \{\otimes\}$.

For a di-graph $G=(V,E)$ and a subset $W\subseteq V$ we denote with $G[W] = (W,F)$ 
the \emph{induced subgraph} of $G$, i.e., any arc $xy\in E$ with $x,y\in W$
is also contained in $G[W]$.

In what follows,  $\irr{X\times X}$ denotes the set 
$(X\times X)\setminus \{(x,x)\mid x\in X\}$. 
To avoid trivial cases, we always assume that $|X|>1$. 
The sets $X_1,\dots,X_k$ form a \emph{quasi-partition} of $X$, 
if all sets are pairwisely disjoint, their union is $X$ and 
at most one $X_i$ is empty.

%By abuse of notation, we say that $X_1,\dots,X_k$ forms a \emph{partition} of $X$, 
%if all sets are pairwisely disjoint, their union is $X$ and 
%at most one $X_i$ is empty.

\subsection{Simple Fitch Graphs}

Let $\lambda:E\to \{1,\otimes\}$ be a map and $(T,\lambda)$ be an edge-labeled phylogenetic tree on $X$. 
We set $(x,y)\in\mathcal{X}_{(T,\lambda)}$ for $x,y\in X$ 
whenever the uniquely defined path from
$\lca_T(x,y)$ to $y$ contains at least one 1-edge. 
By construction $\mathcal{X}_{(T,\lambda)}$ is irreflexive; hence it can be regarded
as a simple directed graph.

An arbitrary di-graph $G=(X,E)$ is \emph{explained} by a phylogenetic tree $(T,\lambda)$ (on $X$) and called 
\emph{simple Fitch graph}, whenever $xy\in E$ if and only if $(x,y)\in \mathcal{X}_{(T,\lambda)}$. 
Fitch graphs are the topic of Ref.\ \cite{Geiss:17}, which among other results gave
a characterization in terms of eight forbidden induced subgraphs. The following theorem
summarizes a couple of important results that we need for later reference.

\begin{theorem}[\cite{Geiss:17}]
A given di-graph $G=(X,E)$ is a simple Fitch graph if and only if it does not contain
one the graphs $F_1,\dots, F_8$ (shown in Fig.\ \ref{fig:forb}) as an
induced subgraph. 

Deciding whether $G$ is a simple Fitch graph and, in the positive
case, to construct the unique least-resolved tree $(T,\lambda)$ that explains $G$ can be done
in $O(|X|+|E|)$ time.  

$(T,\lambda)$ is a least-resolved tree that explains $G$, i.e., 
there is no edge-contracted version $T'$ of $T$ and no labeling $\lambda'$
such that $(T',\lambda')$ still explains $G$, 
if and only if 
all its inner edges are 1-edges and for every inner edge $(u,v)$ 
there is an outer $\otimes$-edge $(v,x)$ in   $(T,\lambda)$.
\label{thm:1-Fitch}
\end{theorem}

\begin{figure}[tbp]
\begin{center}
  \includegraphics[width=0.6\textwidth]{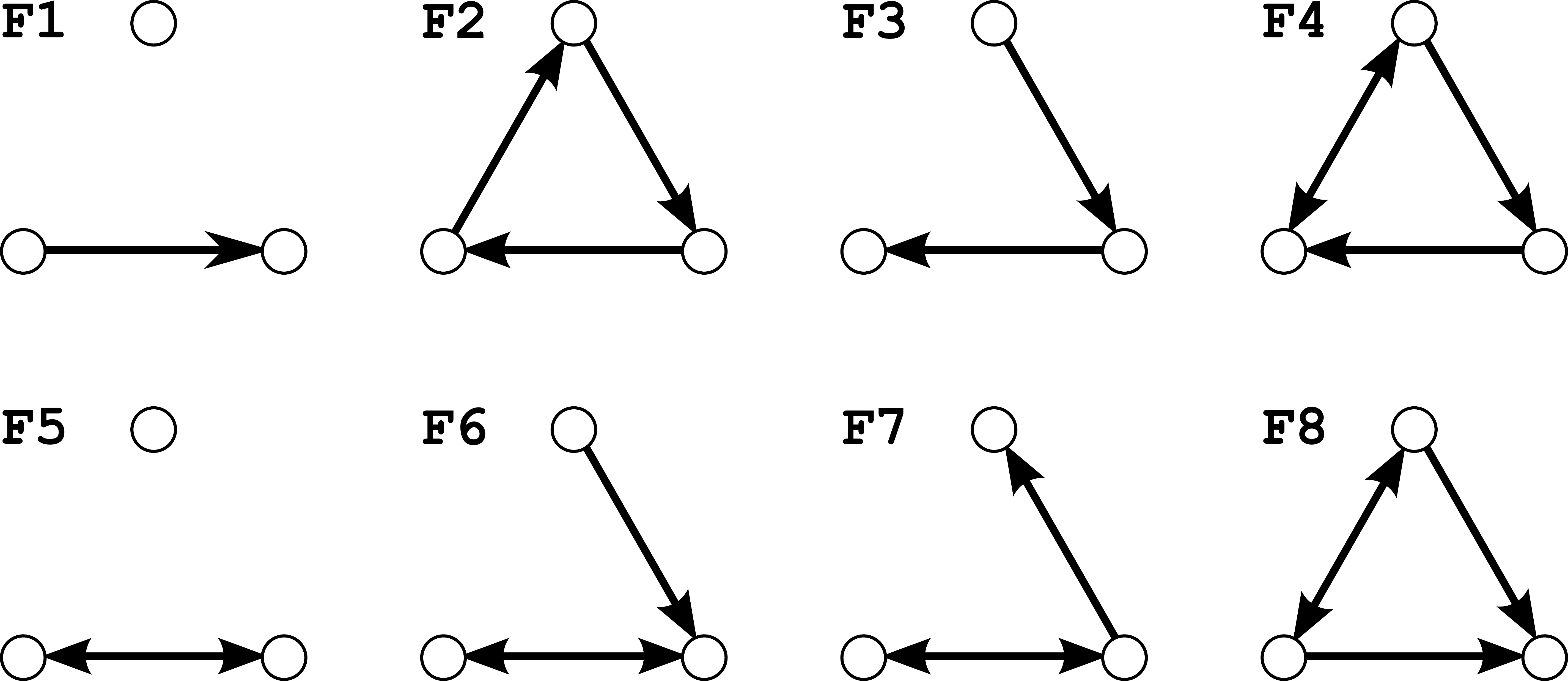}
\end{center}
\caption{Shown are the eight forbidden induced subgraphs $F_1,\dots,F_8$ of Fitch graphs.}
\label{fig:forb}
\end{figure}

\section{Generalized Fitch Graphs}

To generalize the notion of simple Fitch graphs, we
consider complete di-graphs $(K_{|X|},\eps)$ with vertex $X$, arc set $\irr{X\times X}$ and 
a map $\eps:\irr{X\times X}\to \Mo$ that assigns to each arc $xy$ a unique label
$\eps(x,y)$. Clearly, the map $\eps$ covers all information provided by $(K_{|X|},\eps)$. 
W.l.o.g.\ we will always assume that for each $m\in M$ there is at least one pair  $(x,y)\in \irr{X\times X}$
such that $\eps(x,y) = m$. 

\begin{definition}
Let $\eps:\irr{X\times X}\to \Mo$ be a map. For a given phylogenetic tree $(T,\lambda)$ with $\lambda:E\to \Mo$ and
two leaves $x$ and $y$ we denote with $\mathbb{P}_{(x,y)}$ the unique path in $T$ from 
$\lca_T(x,y)$ to $y$. 
A pair $(x,y)\in \irr{X\times X}$ is \emph{explained by} a phylogenetic tree $(T,\lambda)$ on $X$
whenever, 
\begin{description}
	\item  $\eps(x,y) = m\in M$ iff some edge $e$ on the path $\mathbb{P}_{(x,y)}$ has 
					label $\lambda(e)=m$; and 
	\item $\eps(x,y) = \otimes$ iff none of the edges $e$ on the path $\mathbb{P}_{(x,y)}$
				have label $\lambda(e)\in M$. 
 \end{description}
  
The map $\eps$ is \emph{tree-like} if each pair $(x,y)\in \irr{X\times X}$ is explained by $(T,\lambda)$. 
In this case,  we say  that $(T,\lambda)$ \emph{explains} $\eps$
and $(K_{|X|},\eps)$ is a \emph{(generalized) Fitch graph}. 

Moreover, a tree $(T,\lambda)$ is \emph{least-resolved} for a map $\eps$, 
if $(T,\lambda)$ explains $\eps$ and there is no tree $(T',\lambda')$
that explains $\eps$, where $T'$ is obtained from $T$ by contracting
edges and $\lambda'$ is an $M\cup \{\otimes\}$-edge-labeling map. 
\end{definition}

Figure \ref{fig:exmpl} shows an example of a generalized Fitch graph $(K_{|X|},\eps)$. 
We give the following almost trivial result for later reference. 

\begin{lemma}
   Let $\eps:\irr{X\times X}\to \Mo$ be tree-like and  	$(T,\lambda)$ be a tree that explains 
	 $\eps$. If there is an edge $e$ with $\lambda(e)=m$ on the path $P$ from the root $\rho_T$ to some leaf, then all edges in $P$ are either labeled with $m$ or $\otimes$.  	
\label{lem:path-one-color}
\end{lemma}
\begin{proof}
	Let $P$ be the path from the root $\rho_T$ to the leaf $x\in X$.
	Let $v$ be the child of $\rho_T$ that is an ancestor of $x$. 
	Now let $y\in X$ be any leaf that is not a descendant of $v$ and 
	thus $\lca_T(x,y) = \rho_T$.  
	Assume, for contradiction, that there are two edges in $P$ with distinct labels 
	$m,m'\in M$. 	
	Since  $(T,\lambda)$ explains $\eps$ we would have $\eps(y,x) = m$ and
	$\eps(y,x) = m'$; a contradiction to $\eps$ being a map. 
 \end{proof}

For each symbol $s\in \Mo$ we define the following set
%	\[	X_s \coloneqq \{x\in X \mid \text{ there is a vertex } z\in X \text{ with } \eps(z,x) = s 
%																\text{ and for all }	z'\in X\setminus\{z,x\} \text { we have }  \eps(z',x)\in \{\otimes, s\} \}
%	\]
\begin{align*}
		X_s \coloneqq \{x\in X \mid &\text{ there is a vertex } z\in X \text{ with } \eps(z,x) = s \\
																&\text{ and for all }	z'\in X\setminus\{z,x\} \text { we have }  \eps(z',x)\in \{\otimes, s\} \}
\end{align*}
that contains for each symbol $s$ those vertices $x\in X$ where at least one incoming arc is labeled $s$
and all other incoming arcs have label $s$ or $\otimes$. 
Note, by construction for all $x,y\in X_{\otimes}$
 we have $\eps(x,y) = \eps(y,x)= \otimes$
and for all $x,y\in X_m$, $m\in M$
 we have  $\eps(x,y), \eps(y,x) \in \{m, \otimes\}$.

The intuition behind the sets $X_s$ is sketched in Fig.\ \ref{fig:intuition}. 
In this example, let $(K_{|X|},\eps)$ be the Fitch graph that is explained by the sketched tree and 
assume that the highlighted $m$-edge $e$ with $m\neq \otimes$
is the first $m$-edge that lies on the path from the root to any of the leaves that
are located below this edge. Lemma \ref{lem:path-one-color} implies that
all edges on this path that are above $e$ must be $\otimes$-edges and all
edges below $e$ must either be $\otimes$- or $m$-edges. 
This observation implies that every leaf $z$ located in the subtree $T'$
must ``point to'' to every leaf $x \in X'_m$ via an $m$-edge in $(K_{|X|},\eps)$, i.e., $\eps(z,x)=m$. 
Moreover, for any two vertices $z',x\in X'_m$ we have $\eps(z',x) \in  \{\otimes, m\}$.
Thus, the set $X'_m$ in Fig. \ref{fig:intuition} is a subset of the (possibly 
larger) set $X_m$.
\begin{figure}[tbp]
\begin{center}
  \includegraphics[width=0.4\textwidth]{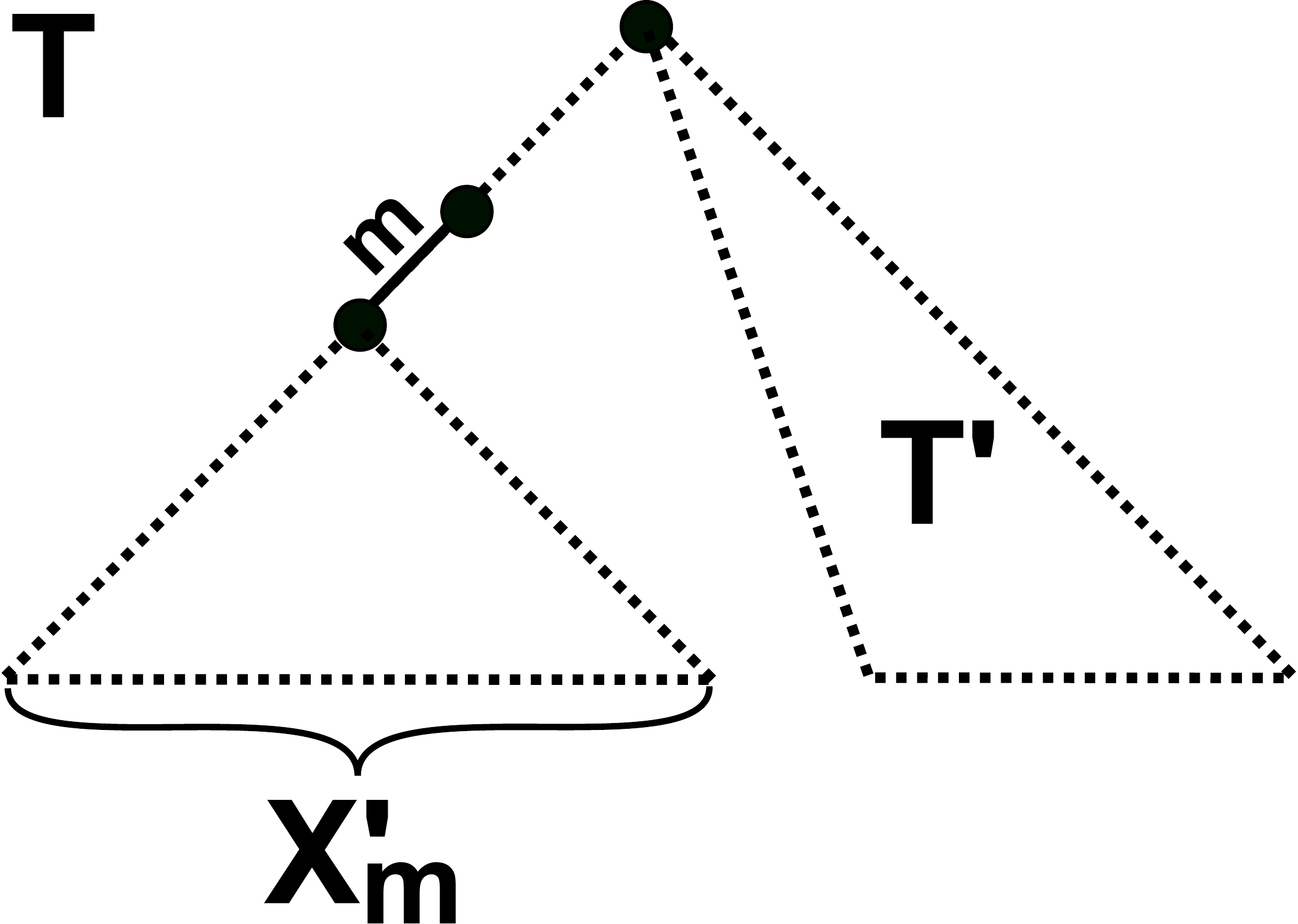}
\end{center}
\caption{Shown is a sketched tree $T$, where the highlighted $m$-edge $e$ is the first $m$-edge
					that is located on the path from the root to any of the leaves below $e$. 
					Thus, $X'_m\subseteq X_m$, see text for further details.}
\label{fig:intuition}
\end{figure}

\begin{lemma}
	Let $\eps:\irr{X\times X} \to \Mo$  be a tree-like map
	and $(T,\lambda)$ a tree that explains $\eps$. 
	Then,  for all $m\in M$ we have $X_m\neq \emptyset$
	and $X_m =  \{x\in X \mid \exists z\in X \text{ with } \eps(z,x) = m\}$. 
	In particular, 
	the sets $X_1,X_2,\dots X_{|M|},X_{\otimes}$ form a quasi-partition of $X$.

	Moreover, for all $x\in X_m$ and $y\in X\setminus X_{m}$ with $m\in \Mo$ 
	it holds that $\eps(y,x) = m$.
\label{lem:basics}
\end{lemma}
\begin{proof}	
	To recap, $\eps$ is a map such that $\eps^{-1}(m)\neq \emptyset$ for all $m\in M$. 
	Thus, for each $m\in M$ there are two vertices $x,z\in X$ with $\eps(z,x)=m$. 
	Assume for contradiction that there is a vertex $z'\in X$ with $\eps(z',x)=m'\notin \{m,\otimes\}$.
	Thus, the path from $\lca(z,x)$  to $x$ contains an edge labeled $m$
	and the path from $\lca(z',x)$  to $x$ contains an edge labeled $m'$. 
	However, since both vertices $\lca(x,z)$ and $\lca(x,z')$ are located on the path from
	the root of $T$ to $x$, this path must
	have two edges, one with label $m$ and one with label $m'$; a contradiction
	to Lemma \ref{lem:path-one-color}. Thus, $\eps(z,x)\in \{m,\otimes\}$ for all $z\in X\setminus\{x\}$ and therefore, 
	$x\in X_m$. Thus,  $X_m\neq \emptyset$ for all $m\in M$.
	In particular,  the latter arguments imply that 
	whenever there  are vertices $x,z\in X$ with $\eps(z,x)=m$, 
	then for all vertices $z'\in X\setminus\{x\}$ we have 
	 $\eps(z',x)\in \{m,\otimes\}$ and thus, the sets  
	$X_m$ and $\{x\in X \mid \exists z\in X \text{ with } \eps(z,x) = m\}$
	are identical.

	\smallskip
	We continue to show that $X_1,X_2,\dots X_{|M|},X_{\otimes}$ form a quasi-partition of $X$.
	Clearly, for all distinct $m,m' \in M$ the sets $X_m, X_{m'} $ must be disjoint, 
	as otherwise $x\in X_m\cap X_{m'}$ would imply that 
	$\eps(z,x)=m$ for some $z\in X$ and at the same time 
	$\eps(z,x)\in \{m',\otimes\}$; a contradiction to $\eps$ being a map. 
	Moreover, for all distinct $m\in M$ the sets $X_m, X_{\otimes} $ must be disjoint,
	since $x\in X_{\otimes}$ if and only if $\eps(z,x)=\otimes$ for all $z\in X\setminus\{x\}$, 
	which is, if and only if  $x\notin X_m$ for all $m\in M$.
	
	It remains to show that the union of $X_1,X_2,\dots X_{|M|},X_{\otimes}$
	is $X$ and at most one of the sets is empty.  
	Note, for each $m\in M$ there are two vertices $z,x\in X$ with $\eps(z,x)=m$. 
	As argued above, $\eps(z,x)=m\in M$ implies $x\in X_m$.
 	Thus, none of the sets $X_1,X_2,\dots X_{|M|}$ is empty. 
	In particular,	$X_{\otimes} = \emptyset$  if and only if 
	for every $x\in X$ we have  $\eps(z,x)=m$ for some $z\in X$ and $m\in M$. 
	In this case, the union of the sets $X_1,X_2,\dots X_{|M|}$
	is $X$. 
	Now assume that $X_{\otimes} \neq \emptyset$ and $x\notin X_m, m\in M$. 
	Hence, 	$\eps(z,x)\neq m$ for all $z\in X\setminus\{x\}$ and all $m\in M$. 
	Thus, $\eps(z,x) = \otimes$ for all $z\in X\setminus\{x\}$, and therefore, 
	$x\in X_{\otimes}$. Thus, in case $X_{\otimes} \neq \emptyset$, 
	the union of the sets 	$X_1,X_2,\dots X_{|M|},X_{\otimes}$ is $X$.

	\smallskip
	To prove the last statement, let $x\in X_m$. 
	Clearly, if $m=\otimes$ and thus, $x\in X_{\otimes}$ then
	$\eps(y,x) = \otimes$ for all $y\in X\setminus\{x\}$.  
	Now, let $m\in M$ and $m'\in \Mo$ 
	with $m\neq m'$. Assume for contradiction  that $\eps(y,x) \neq m$
	for some $y\in X_{m'}$.
	Thus, the path from $\lca(x,y)$ to $x$ does not 
	contain an $m$-edge. By construction of $X_m$, 
	there is a vertex $z\in X$ with $\eps(z,x)=m$	and thus, 
	the path from $\lca(x,z)$ to $x$ contains an $m$-edge $e=uv$. 
	Trivially, all ancestors of $x$ are located on 
	the path from the root of $T$ to $x$ and thus, also $\lca(x,z)$ and $\lca(x,y)$. 
  Therefore, the $m$-edge is located between $\lca(x,z)$ and $\lca(x,y)$ and, 
	in particular,  
	$\lca(x,z)\succeq u\succ v\succeq \lca(x,y)$.
	Hence, $\lca(x,z) = \lca(y,z)$ and the path from 
	$\lca(y,z)$ to $z$ contains an $m$-edge; a contradiction
	to $y\in X_{m'}$.
 \end{proof}

\begin{figure}[tbp]
\begin{center}
  \includegraphics[width=\textwidth]{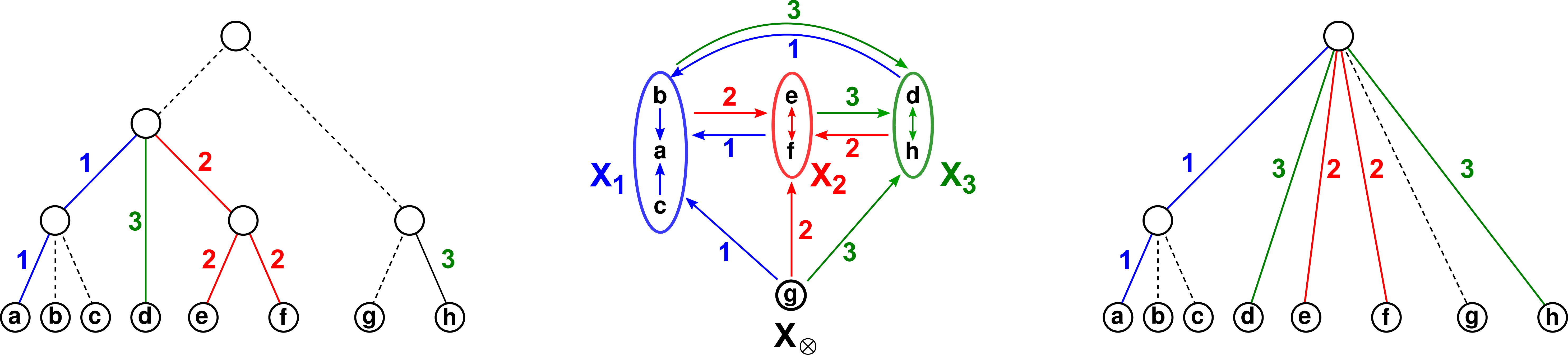}
\end{center}
\caption{Left, an edge-labeled tree $(T,\lambda)$ is shown where $\otimes$-edges are
				 drawn as dashed-lines. The tree $(T,\lambda)$ explains the graph $(K_{|X|},\eps)$ in 
					the middle, where $X= \{a,b,\dots,h\}$ and $\eps: \irr{X\times X} \to \{1,2,3,\otimes\}$.
						For better readability, $\otimes$-edges are omitted in the
						drawing of $(K_{|X|},\eps)$ and only one arc $xy$ with label 
						$\eps(x,y)=m'$ and $\eps(y,x)=m$ for each $x\in X_m$ and $y\in X_{m'}$, $m\neq m'$
					is drawn. All arcs between vertices $x,y\in X_m$ have label $m$. 	
					The least-resolved tree $(\Ts,\ls)$ constructed with Alg.\ \ref{alg:genfitch}
					is shown at the right-hand side.}
\label{fig:exmpl}
\end{figure}

For each $m\in M$, we denote with $G_m$ the subgraph of  $(K_{|X|},\eps)$ with vertex set 
$X_m$ as defined above and arc set \[E_m = \{xy \mid x,y\in X_m, \eps(xy)\neq \otimes\}.\]
Note, by definition of $X_m$, the graph $G_m$ contains only arcs $xy$ with $\eps(xy)=m$. 

Before we can derive the final result, we need one further definition. 
Let $(T,\lambda)$ be an edge-labeled phylogenetic tree on $X$. 
To recap,  the restriction $T_{|X_m}$ of $T$ to $X_m$ 
is obtained by suppressing all degree-$2$ vertices of $T(X_m)$.
For any edge $uv\in E(T_{|X_m})$,  let $S(u,v)$ denote the set
of all suppressed vertices on the path from $u$ to $v$ in $T(X_m)$.
We define the restriction $\lambda_{|X_m}$ to 
$X_m$ by putting for all edges
$uv\in E(T_{|X_m})$: 
	\[\lambda_{|X_m}(x,y) = \begin{cases}
			\lambda(u,v) & \, \text{, if } S(u,v)=\emptyset \text{ and thus, } uv\in E(T) \\
			 m & \, \text{, else if } \text{there are } a,b\in S(u,v)\cup\{u,v\} \text{ with } \lambda(a,b)=m\\
			 \otimes & \, \text{, else.}
	\end{cases}\]
Lemma \ref{lem:path-one-color} implies that the restriction $\lambda_{|X_m}$ of $\lambda$
is well defined. In particular, $\lambda_{|X_m}(u,v) = m$ if and only if the corresponding unique path
between $u$ and $v$ in $T$ contains an $m$-edge.

We are now in the position to characterize tree-like maps $\eps$.

\begin{theorem}
\label{thm:main}
The map $\eps:\irr{X\times X} \to \Mo$ is tree-like 
(or equivalently $(K_{|X|},\eps)$ is a generalized Fitch graph)
if and only if 
the following four conditions are satisfied:
\begin{description}
\item[(T1)] The sets $X_1,X_2,\dots X_{|M|},X_{\otimes}$ form a quasi-partition of $X$.
\item[(T2)] $G_m  = (X_m,E_m)$ is a simple Fitch graph for all $m\in M$. 
\item[(T3)] For all $m\in M$ and $x\in X_m$, $y\in X\setminus X_{m}$
						it holds that $\eps(y,x) = m$. 
\item[(T4)] For all $x\in X_{\otimes}$ and $y\in X\setminus\{x\}$ 	
						it holds that $\eps(y,x) = \otimes$.
\end{description}
In particular, the tree $(\Ts,\ls)$ returned by Algorithm \ref{alg:genfitch}
(with input $\eps$) explains $\eps$, whenever $\eps$ is tree-like.  

\end{theorem}
\begin{proof}
	We first establish the `if' direction. Assume Conditions \emph{(T1)} to \emph{(T4)}
	are satisfied for $\eps$. Since $G_m = (X_m,E_m)$ is a simple Fitch graph for
	all $m\in M$, all $G_m$ are explained by a 
	tree	$(T_m,\lambda_m)$ with leaf set $X_m$. 
	
	We show that the tree $(\Ts,\ls)$ constructed with Alg.\ \ref{alg:genfitch} explains $\eps$.
	By construction of  $(\Ts,\ls)$ all trees $(T_m,\lambda_m)$ are exactly the subtrees $\Ts(X_m)$ where 
	all edge labels $\lambda_m$ are kept. 
	Hence, $G_m$ is explained by $(\Ts(X_m),\lambda^*_{|X_M})$. 
			Since $\eps(x,y) = m$ (resp.\ $\eps(x,y) = \otimes$) for any $x,y\in X_m$
	if and only if $xy\in E_m$ (resp.\ $xy\not\in E_m$), we can conclude that 
	all pairs $(x,y), (y,x)$ with $x,y\in X_m$ are explained by $(\Ts,\ls)$ for all $m\in M$. 
	Moreover, each $x\in X_{\otimes}$ 	
	is linked to the root $\rho_{\Ts}$ via an $\otimes$-edge (Line \ref{alg:x4} of Alg.\ \ref{alg:genfitch}). 
	Hence, for each two vertices $x,y\in X_{\otimes}$ we have, by definition of $X_{\otimes}$, 
	$\eps(x,y) = \eps(y,x) = \otimes$, which is trivially explained by $(\Ts,\ls)$. 
	Since the sets $X_1,X_2,\dots X_{|M|},X_{\otimes}$ 
	form a quasi-partition of $X$, it is ensured that there are  no overlapping leaf sets when
	the trees $(T_m,\lambda_m)$ and the elements $x\in X_{\otimes}$ 
	have been added to $(\Ts,\ls)$ and that the leaf set of $\Ts$ is $X$. 

	We continue to show that all pairs $(x,y)$ with $x,y\in X$ that satisfy \emph{(T3)} and \emph{(T4)} are
	explained by $(\Ts,\ls)$. 
	Note first, by construction of  $(\Ts,\ls)$ and since $(\Ts,\ls)$	
	explains $G_m$ for all $m\in M$, all edges along the 
	path from $\rho_{\Ts}$ to $x\in X_m$ have label $m$ or $\otimes$.
	Even more, we show that each path $P_{\rho_{\Ts},x}$ from $\rho_{\Ts}$ to each $x\in X_m$, $m\in M$
	has always an edge with label $m$. 
 	By construction of $(\Ts,\ls)$ (Alg.\ \ref{alg:genfitch}, Line \ref{alg:x1}-\ref{alg:x3}), if
	$|X_m|=1$ or there is a leaf  $x\in X_m$ adjacent to the root $\rho_m$ of $T_m$
	such that $\lambda_m(\rho_m,x)=\otimes$, then the tree $(T_m,\lambda_m)$
	is placed below the particular $m$-edge $\rho_{\Ts}r_m$. Hence, all paths from
 	$\rho_{\Ts}$ to $x\in X_m$ contain this $m$-edge. Since \emph{(T3)} and \emph{(T4)} state that 
	$\eps(y,x)=m$ for all $x\in X_m$, $m\in M$, all pairs $(y,x)$ with
	$x\in X_m$, $y\in X\setminus X_m$ are explained by $(\Ts,\ls)$, given that 
	$(T_m,\lambda_m)$ satisfies the Conditions in Alg.\ \ref{alg:genfitch} (Line \ref{alg:x1}). 
	Assume that $(T_m,\lambda_m)$ does not satisfy the latter conditions. 
 	Theorem \ref{thm:1-Fitch}	
	implies that all inner edges of $(T_m,\lambda_m)$ are $m$-edges and thus, any $\otimes$-edge in 
	$(T_m,\lambda_m)$ must be incident to some leaf $x\in X_m$. 
  Since $(T_m,\lambda_m)$ does not satisfy the if-condition in Line \ref{alg:x1} of  Alg.\ \ref{alg:genfitch}, 
	\emph{all} edges that are incident to the root of $(T_m,\lambda_m)$ 
	have label $m$. Hence, all paths from
 	$\rho_{\Ts}$ to $x\in X_m$ contain  an $m$-edge and, 
	therefore, all pairs $(y,x)$ with
	$x\in X_m$, $y\in X\setminus X_m$ are explained by $(\Ts,\ls)$.  
	Finally, if $x\in X_{\otimes}$, then \emph{(T4)} claims $\eps(y,x)=\otimes$ for all $y\neq x$
	which is trivially explained by $(\Ts,\ls)$, since 
	$x$ is linked to the root $\rho_{\Ts}$ via an $\otimes$-edge
	(Alg.\ \ref{alg:genfitch}, Line \ref{alg:x4}). 
	In summary, if the Conditions \emph{(T1)} to \emph{(T4)}
	are satisfied, then $\eps$ is explained by $(\Ts,\ls)$ and therefore, 
	tree-like.  
	This establishes the `if' direction.

	We turn now to the `only if' direction. Assume that $\eps$ is tree-like 
	and let $(T,\lambda)$ be a tree that explains $\eps$ with root $\rho_T$. 
	Lemma \ref{lem:basics} implies Condition \emph{(T1)}, \emph{(T3)} and \emph{(T4)}.
		We continue to show \emph{(T2)}. To this end,   	 
	consider the graph $G_m = (X_m,E_m)$, $m\in M$. 
	Since $(T,\lambda)$ explains $\eps$  and therefore, also $(K_{|X|},\eps)$, it must 
	explain each of its induced subgraphs and thus, 
	any pair $(x,y)$ with $x,y\in X_m$ and $m\in M$ is explained 
	by $(T,\lambda)$.  
	By construction of the restriction $(T_{|X_m},\lambda_{|X_m})$
	of $(T,\lambda)$ to $X_m$ we have 
	$\eps(x,y)=m$ if and only if
	the path in $(T, \lambda )$ from $\lca_T(x,y)$ to $y$ contains an 
	$m$-edge which is  
	if and only if  there is an
	$m$-edge on the path from $\lca_{T_{|X_m}}(x, y)$ to the leaf $y$ in $(T_{|X_m}, \lambda_{|X_m})$.  
	Hence,  $(T_{|X_m}, \lambda_{|X_m})$ explains $(K_{|X|}[X_M],\eps)$. 
	By definition of $X_m$, 
  the graph $G_m$ contains only arcs $xy$ with $\eps(x,y)=m$
	and  for all $x,y\in X_m$ with $xy\notin E_m$
	we have $\eps(x,y)=\otimes$. 
	Thus, $G_m$ is obtained from $(K_{|X|}[X_M],\eps)$ by removing
	all $\otimes$-edges and is, therefore, explained by 
	 $(T_{|X_m}, \lambda_{|X_m})$. Hence, $G_m$ is a simple Fitch graph and \emph{(T2)}
	is satisfied. 	This establishes the `only if' direction.

	Thus, Conditions \emph{(T1)} to \emph{(T4)} characterize 
	 tree-like maps $\eps$. This together with the 
	proof of the `if' direction implies 
	the correctness of Alg.\ \ref{alg:genfitch}. 
 \end{proof}

\begin{algorithm}[tbp]
\caption{\texttt{Recognition of tree-like maps $\eps$.}} 
\label{alg:recfitch}
\begin{algorithmic}[1] %\small
  \Require A map $\eps:\irr{X\times X} \to \Mo$;  
  \Ensure A least-resolved edge-labeled tree $(\Ts,\ls)$ that explains $\eps$ or the statement ``The map $\eps$ is not tree-like'';
	\If{$|M|>2|X|-2$}  			 Output: ``The map $\eps$ is not tree-like''; \label{alg:Msize}
	\ElsIf{ $\eps$ satisfies Condition \emph{(T1)} to \emph{(T4)} in Thm.\ \ref{thm:main} }\label{alg:y1}
	\State Compute $(\Ts,\ls)$ with Alg.\ \ref{alg:genfitch};
			\Else\  
			 Output: ``The map $\eps$ is not tree-like'';
     \EndIf
\end{algorithmic}
\end{algorithm}

\begin{algorithm}[tbp]
\caption{\texttt{Compute Least-Resolved Tree $(\Ts,\ls)$ for $\eps$.}} 
\label{alg:genfitch}
\begin{algorithmic}[1] %\small
  \Require A tree-like map $\eps:\irr{X\times X} \to \Mo$; 
  \Ensure A least-resolved edge-labeled tree $(\Ts,\ls)$ that explains $\eps$;
	\State Add a root $\rho_{\Ts}$ to $\Ts$;
  \For{all $m\in M$}
	\State	Compute the least-resolved tree $(T_m,\lambda_m)$ the explains $G_m=(X_m,E_m)$; 	
     \If{$|X_m|=1$ \textit{OR} $(T_m,\lambda_m)$ contains an $\otimes$-edge incident to its root}  \label{alg:x1}
	\State		Add a vertex $r_m$ and the edge $\rho_{\Ts} r_m$ with label $m$; \label{alg:x2}
	\State		Add $(T_m,\lambda_m)$ by identifying the root of $T_m$ with $r_m$; \label{alg:x3}
  \State    Set $\lambda^*(e) = \lambda_m(e)$ for all edges in $T_m$;	
		\Else\ 
			 Identify the root of $T_m$ with $\rho_{\Ts}$ and add	$(T_m,\lambda_m)$;
     \EndIf
  \EndFor
	\State Add an edge  $e = \rho_{\Ts} x$ with label $\ls(e)=\otimes$ for all $x\in X_{\otimes}$; \label{alg:x4}
	\State Return $(\Ts,\ls)$;
	\end{algorithmic}
\end{algorithm}

\begin{theorem}
	For a given map $\eps:\irr{X\times X} \to \Mo$, 
	Algorithm \ref{alg:recfitch} determines whether 
	$\eps$ is tree-like or not, and returns a tree $(\Ts,\ls)$
	that explains a tree-like map $\eps$ in  $O(|X|^2)$-time.
	
	In particular, if $\eps$ is tree-like, then $(\Ts,\ls)$
  is a least-resolved tree for $\eps$.
\end{theorem}
\begin{proof}
To establish the correctness of Alg.\  \ref{alg:recfitch}, 
note first that for any tree $T=(V,E)$ on $X$ we have
$|E|+1=|V| \leq 2|X|-1$ (cf.\ \cite[Lemma 1]{Hellmuth:15a}). 
Thus, there is no tree with $|E| > 2|X|-2$ edges and hence, 
one can place at most $2|X|-1$ different symbols on the
edges of a tree. Therefore, 
if $|M|>2|X|-2$, then $\eps$ cannot be tree-like, 
since we claimed that for any $m\in M$, 
$\eps^{-1}(m)\neq \emptyset$. This establishes
the correctness of Line \ref{alg:Msize} of Alg.\ \ref{alg:recfitch}. 
Now, apply Thm.\ \ref{thm:main} to conclude that 
Alg.\ \ref{alg:recfitch} is correct. 

We continue to verify the runtime of Alg.\  \ref{alg:recfitch}. 
Clearly, the sets $X_1,\dots,X_{|M|}, X_{\otimes}$ can be constructed
by stepwisely considering each pair $(x,y)\in \irr{X\times X}$ and 
its label $\eps(x,y)$, which takes  $O(|X|^2)$-time. 
In particular, verifying Condition \emph{(T1)} 
can be done directly within the construction phase of the sets 
$X_m$, $m\in \Mo$ and, hence stays within the time complexity 
of $O(|X|^2)$. 
Thm.\ \ref{thm:1-Fitch} implies that Condition \emph{(T2)}
can be verified in $O(|X|+|E_m|)$ time for each $m\in M$. 
Due to the `if-condition' in Line \ref{alg:Msize} of Alg.\ \ref{alg:recfitch}, 
we have $|M|\in O(|X|)$. 
Furthermore, $\sum_{m\in M} E_m \in O(|X|^2)$. 
Thus, Condition  \emph{(T2)} can be checked in
$\sum_{m\in M} O(|X|+|E_m|) =  O(|M||X|) +|X|^2) = O(|X|^2)$ time. 
Finally, for \emph{(T3)} and  \emph{(T4)} we need to check if for
all $x\in X_m$ and $y\in X\setminus X_m$ it holds that $\eps(y,x) = m$.
In other words, we must check for all $x\in X$ which label its
$|X|-1$ incoming  arcs $zx$ have.  
This can be done in $O(|X|^2)$-time. Thus, we end in overall
time-complexity of $O(|X|^2)$ for Alg.\ \ref{alg:recfitch}.

We continue to show that $(\Ts,\ls)$ is a least-resolved tree for $\eps$. 
By construction of  $(\Ts,\ls)$ all trees $(T_m,\lambda_m)$ 
are exactly the subtrees $\Ts(X_m)$ where all edge labels $\lambda_m$ are kept.
Hence, $(T_m,\lambda_m) = (\Ts(X_m),\lambda^*_{|X_M})$.
Note that none of the edges can be contracted that are contained in any of the
trees $(T_m,\lambda_m)$ that explains $G_m$  
and thus,  that explains also
any pair $(x,y)$ with $x,y\in X_m$, 
since $(T_m,\lambda_m)$ is already the unique 
least-resolved for the map $\eps$ restricted to pairs $(x,y)$ with 
$x,y\in X_m$ (cf.\ Thm.\ \ref{thm:1-Fitch}). 
In particular, Thm.\ \ref{thm:1-Fitch} implies that
the labeling $\lambda_m$ is unique and can therefore, 
not be changed. 
Moreover, no outer-edge of $(\Ts,\ls)$ can be contracted, 
otherwise we would loose the information of a leaf. 
Hence, the only remaining edges that might be contracted are the
$m$-edges of the form  $\rho_{\Ts}r_m$ as constructed in Line \ref{alg:x2} of Alg.\ \ref{alg:genfitch}. 
However, such an edge $\rho_{\Ts}r_m$ was only added if 
$(T_m,\lambda_m)$ contains an outer $\otimes$-edge $r_m x$ where $x\in X_m$  and 
 $r_m$ denotes the root  of $T_m$. Thus, contracting the edge $\rho_{\Ts}r_m$
would yield $\rho_{\Ts}=r_m$. Now, there are two possibilities, either 
we relabel the resulting edge $\rho_{\Ts}x$ or we keep the label $\otimes$. 
However, relabeling of $\rho_{\Ts}x$  is not possible,
since  $\lambda_m$ is unique and can therefore, not be changed. 
Thus, $\rho_{\Ts}x$ must remain an $\otimes$-edge. 
However, due to the definition of $X_m$
there is a pair $(z,x)$ with $\eps(z,x)=m$ which cannot be explained
by any tree where $x$ is linked to the root $\rho_{\Ts}$ via an $\otimes$-edge;
a contradiction. Hence, $m$-edges of the form  $\rho_{\Ts}r_m$
cannot be contracted.  In summary, 
there is no tree $(T',\lambda')$
that explains $\eps$, where $T'$ is obtained from $T$ by contracting
an arbitrary edge. Hence, $(\Ts,\ls)$ is least-resolved for $\eps$. 
 \end{proof}

For maps $\eps:\irr{X\times X} \to M$ that assign to 
none of the elements $(x,y)$  a label $\otimes$ we obtain 
the following result.
\begin{corollary}
The map $\eps:\irr{X\times X} \to M$ is tree-like if and only if 
Condition \emph{(T1)} and  \emph{(T3)} are satisfied.  
%A  least-resolved tree $(\Ts,\ls)$ that explains $\eps$ is a star, i.e., 
%a tree on $X$ that does not contain inner vertices except the root. 
%Thus, $(\Ts,\ls)$ has the minimum
%number of vertices among all trees that explain $\eps$. 
\end{corollary}
\begin{proof}
 By Thm.\ \ref{thm:main}, \emph{(T1)} and \emph{(T3)} are satisfied if 
 $\eps$ is tree-like. 
 Assume that 	\emph{(T1)} and \emph{(T3)} are satisfied for $\eps$. 
 By construction of $X_m$, for all $x,y\in X_m$ we have
 $\eps(x,y) = \eps(y,x) = m$. Therefore, 
 $K_{|X|}[X_m] = G_m$ is a complete di-graph with vertex set $X_m$. 
 Hence, $G_m$ does not contain any of the forbidden 
	subgraphs $F_1,\dots,F_8$ (cf.\ Fig.\ \ref{fig:forb}).  
	Therefore, $G_m$ is a simple Fitch graph and \emph{(T2)} is always satisfied. 
	Now, apply Thm.\ \ref{thm:main} to conclude that 
	$\eps$ is tree-like. 
%	Finally, $(K_{|X|}[X_m],\eps)$ is a complete di-graph where
%	all arcs have label $m$. The star tree $(T_m,\lambda_m)$
%	with leaf set $X_m$ where all edges have label $m$ explains $(K_{|X|}[X_m],\eps)$. 
%	Now, we can identify the roots of these star trees
%	$(T_m,\lambda_m)$, $m\in M$  to obtain a tree $(\Ts,\ls)$. 
%	Utilizing \emph{(T1)}	to \emph{(T4)} immediately implies that 
%  $\eps$ is explained by $(\Ts,\ls)$  and, moreover, that
%	no edge can be contracted to obtain a tree $T'$ such
%	that $(T',\lambda')$ still explains $\eps$. 
%	Thus, $(\Ts,\ls)$ is least-resolved, and since it is a star, 
%	it has is necessarily the minimum number of vertices. 
 \end{proof}

\subsection{Uniqueness of the Least-Resolved Tree}

In general, there may be more than one rooted (phylogenetic) tree that explains
a given map $\eps$, see Fig.\ \ref{fig:exmpl}. In particular, if $\eps$ is
explained by a non-binary tree $(T,\lambda)$, then there is always a binary tree
$(T',\lambda')$ that refines $T$ and explains the same map $\eps$ by setting
$\lambda'(e) = \lambda(e)$ for all edges $e$ that are also in $T$ and by
choosing the label $\lambda'(e) = \otimes$ for all edges $e$ that are not
contained in $T$. In this section, we will show that whenever a relation $\eps$
is explained by an edge-labeled tree $(T,\lambda)$, then there exists a unique
least-resolved tree that explains $\eps$. We mainly follow here the proof
strategies as in \cite{Geiss:17}.

To establish the uniqueness of the least-resolved trees, 
we will consider so-called informative triples as shown in 
Fig.\ \ref{fig:inf-triple}. 
Due to Lemma \ref{lem:path-one-color},  it is an easy exercise to verify that 
each edge-labeled graph $G_i$, $i\in \{1,\dots,6\}$ in  Fig.\ \ref{fig:inf-triple}
is explained by the unique edge-labeled binary tree $T_i$, i.e., a specific labeled triple

\begin{figure}[t]
\begin{center}
  \includegraphics[width=0.8\textwidth]{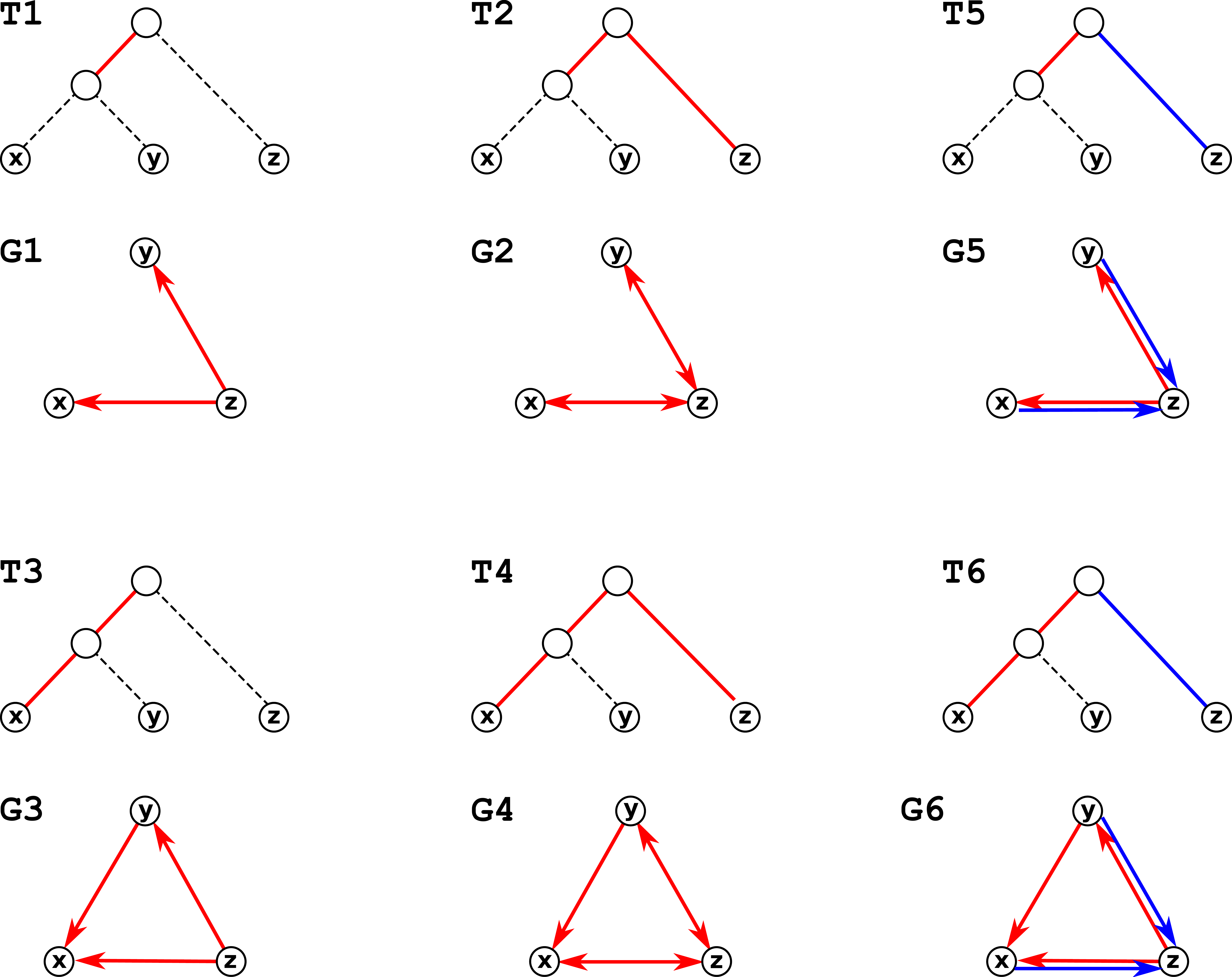}
\end{center}
\caption{Shown is the graph representation for six
			   possible 3-vertex induced edge-labeled subgraphs $G_1,\dots,G_6$ of a generalized Fitch graph$(K_{|X|},\eps)$
				 that is explained by a tree $(T,\lambda)$. 				 The $\otimes$-edges in each graph $G_i$ are omitted. 
				 Each subgraph $G_1,\dots,G_6$ is explained by the
				 unique edge-labeled triple $T_1,\dots,T_6$, respectively. In each tree $T_i$, the $\otimes$-edges are drawn as dashed-lines
					and red-edges and blue-edges correspond to two distinct symbols $m,m'\neq \otimes$. 
				 Edges in $T_1,\dots,T_6$ can be understood as
				  paths in $T$, whereby red-lined (resp. blue-lined, black-dashed) edges 
				  indicate that there is  an $m$-edge (resp. $m'$-edge, only $\otimes$-edges) on the particular path. }
\label{fig:inf-triple}
\end{figure}

\begin{definition} 
  An edge-labeled triple $ab|c$ is \emph{informative} if it explains a
	3-vertex induced subgraphs of a Fitch graph $(K_{|X|},\eps)$
	isomorphic to one of $G_1,\dots G_5$ or $G_6$. 
\end{definition} 
The observation that each graph $G_i$, $i\in \{1,\dots,6\}$ in  Fig.\ \ref{fig:inf-triple}
is explained by the \emph{unique} edge-labeled binary tree $T_i$ is crucial, as this implies that 
whenever $(K_{|X|},\eps)$ contains an induced subgraph of the form  $G_1,\dots G_5$ or $G_6$, 
then any tree explaining $(K_{|X|},\eps)$ must display the corresponding informative triple. 
Any tree-like relation $\eps$ can therefore be
	associated with a uniquely defined set $\IT{\eps}$ of \emph{informative
  triples} that it displays: $r\in \IT{\eps}$ if and only if $r$ is the
	unique edge-labeled triple explaining an induced subgraph isomorphic to
	$G_1,\dots G_5$ or $G_6$. For later reference we summarize this fact as
\begin{lemma}
  If $(T,\lambda)$ explains $\eps$, then all triples in $\IT{\eps}$ must be
  displayed by $(T,\lambda)$.
  \label{lem:display-inform}
\end{lemma}
In what follows, we want to show that $\IT{\eps}$ identifies the least-resolved
tree that explains $\eps$. To this end, we will utilize the following two results. 
\begin{lemma}
	If $(\Ts,\ls)$ is a least-resolved tree for the tree-like map  $\eps:\irr{X\times X} \to \Mo$, 
	then $(\Ts,\ls)$ contains no inner $\otimes$-edges and
	any inner vertex $v\neq \rho_{\Ts}$ of $(\Ts,\ls)$ is incident to an outer $\otimes$-edge.
	\label{lem:inner-edge}
\end{lemma}
\begin{proof}
	First, assume for contradiction that the least-resolved tree 
	$(\Ts,\ls)$ contains an inner $\otimes$-edge $e=uv$. 
	The contraction of the edge $e$ does not change the number
  of $m$-edges with $m\neq \otimes$ along the paths connecting any two leaves. 
	It affects the least  common ancestor of $x$ and $y$, if $\lca_T(x,y)=u$ or $\lca_T(x,y)=v$. In
  either case, however, the number of $m$-edges between the $\lca_T(x,y)$ and
  the leaves $x$ and $y$ remains unchanged. Hence, the map $\eps$ can still be explained
	by the tree that is obtained from $(\Ts,\ls)$ after contraction of $e$. 
	Thus, $(\Ts,\ls)$ is not least-resolved; a contradiction. 

	We continue to show that any inner vertex $v$ must be incident to some outer $\otimes$-edge. 
	Let $e = uv$ be the edge in $\Ts$ with $u\succ_{\Ts}v$.
	Let $F$ be the set of edges that are incident to $v$ and distinct from $e$. 
	First assume, for contradiction, that all edges in $F$
	have a label different from $\otimes$. 
	If there are two edges $f,f'\in F$ with distinct labels, then 
	Lemma \ref{lem:path-one-color} implies that $e$ must be an $\otimes$-edge. 
	However, $(\Ts,\ls)$ contains no inner $\otimes$-edges and, 
	hence, all edges in $F$ must have the same label $m\neq \otimes$. 
	In this case, Lemma \ref{lem:path-one-color} implies that 
	the label $\ls(e)$ of $e$ must be $m$ or $\otimes$. 
	In either case, the edge $e$ can be contracted, since 
	every path from $u$ to a leaf contains already an $m$-edge that is incident to $v$. 
	Thus, $(\Ts,\ls)$ is not least-resolved; a contradiction. 
	Therefore, $v$ must be incident to at least one $\otimes$-edge $f\in F$.
	Since 	$(\Ts,\ls)$ contains no inner $\otimes$-edges, the edge $f$ must be an outer-edge.
\end{proof}

\begin{lemma}
	Each inner edge in a least-resolved tree  $(\Ts,\ls)$  for a tree-like map 
	$\eps:\irr{X\times X} \to \Mo$, is distinguished by at least one informative triple
	in $\IT{\eps}$. 
	\label{lem:edge-dist}
\end{lemma}
\begin{proof}
	Consider an arbitrary inner edge $e=uv$ of $\Ts$ with $u\succ_{\Ts} v$. Since $(\Ts,\ls)$ is phylogenetic, 
	there are necessarily leaves  $x$, $y$, and $z$ such that $\lca(x,y) = v$ and $\lca(x,y,z) = u$. 
	In particular, one can choose $y$ such that $vy$ is an outer $\otimes$-edge, since 
	$(\Ts,\ls)$ is least-resolved and due to Lemma \ref{lem:inner-edge}. 
	Moreover, Lemma \ref{lem:inner-edge} implies that $\ls(e) = m\neq \otimes$. 
	Lemma \ref{lem:path-one-color} implies that all edges $f$ that are located in $\Ts$ below
	$e$ must be $\otimes$- or $m$-edges. Thus, there are two exclusive cases for the
	path from $\lca(x,y)$ to $x$: Either the path contains (a)  only $\otimes$-edges
	or (b) at least one $m$-edge. 
	Moreover, the path $P_{u,z}$ from $u$ to $z$ contains either (A) 
	only $\otimes$-edges or (B) an $m$-edge or (C) an $m'$-edge with $m'\neq m,\otimes$.
	Note, Lemma \ref{lem:path-one-color}  implies that 
	in case (A) (resp.\ (B)) all edges in $P_{u,z}$ must be $m$- or $\otimes$-edges
	(resp.\ $m'$- or $\otimes$-edges). 
	Now, the combination of the Cases (a) and (b) with (A), (B) or (C)
	immediately implies that the tree on $\{x,y,z\}$ displayed by $\Ts$ must be one of the trees
	$T_1,\dots,T_5$ or $T_6$ as shown in Fig.\ \ref{fig:inf-triple}. 
	Therefore, 	$xy|z\in \IT{\eps}$. 
	Since $\lca(x,y) = v$ and $\lca(x,y,z) = u$, the edge
  $e$ is by definition distinguished by the triple $xy|z\in \IT{\eps}$.	
\end{proof}

\begin{theorem}	
	Let $\eps:\irr{X\times X} \to \Mo$ be a tree-like map and $(\Ts,\ls)$ be a least-resolved tree that 
	explains $\eps$. Then,
	the set $\IT{\eps}$ identifies $(\Ts,\ls)$ and $\Aho(\IT{\eps}) = \Ts$. 
	In particular,  $(\Ts,\ls)$  	is unique up to isomorphism.
\end{theorem}
\begin{proof}
  We start with showing that $\IT{\eps}$ identifies $\Ts$. 
	If $\IT{\eps}= \emptyset$, then $(\Ts,\ls)$ must be a star tree, i.e., 
	an edge-labeled tree that consists of outer edges only.
	Otherwise, $(\Ts,\ls)$ contains inner edges
	that are, by Lemma \ref{lem:edge-dist}, distinguished by at least one informative rooted triple in
  $\IT{\eps}$, contradicting that $\IT{\eps} = \emptyset$.  Hence,
  $r(\Ts)=\emptyset$, and therefore, $r(\Ts) = \cl(\IT{\eps})$.  Lemma 
	\ref{g1}  implies that $\IT{\eps}$ identifies $(\Ts,\ls)$.	

  In the case $\IT{\eps} \neq \emptyset$, assume for contradiction that $r(\Ts)
  \neq \cl(\IT{\eps})$. By Lemma \ref{lem:display-inform} we have $\IT{\eps}
  \subseteq r(\Ts)$. Isotony of the closure, Theorem 3.1(3) in
  \cite{Bryant97}, ensures $\cl(\IT{\eps}) \subseteq \cl(r(\Ts))=r(\Ts)$. Our
	  assumption therefore implies $\cl(\IT{\eps}) \subsetneq r(\Ts)$, and thus the
  existence of a triple $ab|c \in r(\Ts)\setminus \cl(\IT{\eps})$. In
	  particular, therefore, $ab|c \notin\IT{\eps}$. Note that neither $ac|b$ nor
  $bc|a$ can be contained in $\IT{\eps}$, since $(\Ts,\ls)$ explains $\eps$
	  and, by assumption, already displays the triple $ab|c$. Thus, $\IT{\eps}$
  contains no triples on $\{a,b,c\}$. 	

	Let $u=\lca(a,b,c)$ and $e = uv$ be the edge in $\Ts$ with 
	$u\succ_{\Ts} v \succeq_{\Ts} \lca(a,b)$. By Lemma \ref{lem:inner-edge}, 
	the edge $e$ must be an $m$-edge with $m\neq \otimes$. 
	Let $T_{abc}$ be the subtree of $(T,\lambda)$ with leaves $a,b,c$. 
	Since $e$ is an $m$-edge, Lemma \ref{lem:path-one-color} implies that
	all edges along the paths from $v$ to $a$ and $v$ to $b$
	must be $m$- or $\otimes$-edges. However, 
	since $ab|c \notin\IT{\eps}$, the tree $T_{abc}$ cannot be 
	isomorphic to the subtree $T_1,\dots,T_6$ and thus, 
	both paths from $\lca(a,b)$ to $a$ and $\lca(a,b)$ to $b$ must contain $m$-edges. 
	
	Moreover, Lemma \ref{lem:inner-edge} implies that there must be an outer $\otimes$-edge 
	$f=vd$. By the discussion above, $d\neq a,b$. 
	Thus, the  the subtrees $T_{acd}$ and
  $T_{bcd}$ of $\Ts$ with leaves $a,c,d$ and $b,c,d$, respectively, correspond to
  one the trees $T_3,T_4$ and $T_6$. 
	By construction, $ad|c\in \IT{\eps}$ and
  $bd|c\in\IT{\eps}$.  Hence, any tree that explains $\eps$ must display $ad|c$
  and $bd|c$.  As shown in \cite{Dekker86}, a tree displaying $ad|c$ and
  $bd|c$ also displays $ab|c$.  This implies, however, that $ab|c \in
  \cl(\IT{\eps})$, a contradiction to our assumption.
	
	Therefore, $\cl(\IT{\eps})=r(T)$ and we can apply Lemma \ref{g1} to
  conclude that $\IT{\eps}$ identifies $(\Ts,\ls)$ and $\Aho(\IT{\eps}) = \Ts$.

	We continue to show the uniqueness of $(\Ts,\ls)$.
	Since  $\IT{\eps}$ identifies $(\Ts,\ls)$, 
	any tree that displays $\IT{\eps}$ is by definition a  refinement of $(\Ts,\ls)$.
	In addition, any tree that explains $\eps$ must display $\IT{\eps}$
	(cf.\ Lemma \ref{lem:display-inform}). 
	Taken the latter two arguments together, any tree that explains $\eps$ must be a refinement
	of $(\Ts,\ls)$. 

	To establish uniqueness of $(\Ts,\ls)$ it remains to show that	
	there is no other labeling $\lambda$ such that $(\Ts,\lambda)$
	still explains $\eps$. 
	Let $e= uv$ be an outer edge. Hence,  changing the label of $e$ would 
	immediately change the label $\eps(w,v)$ between
  $v$ and any leaf $w$ located in a subtree rooted at a sibling of $v$.
  Since at least one such leaf $w$ exists in a phylogenetic tree, the edge
	$e$ 	cannot be re-labeled.
	Now suppose that $e= uv$ is an inner edge with $u\succ_{\Ts} v$. 
	By Lemma \ref{lem:inner-edge}, the edge $e$ must  be $m$-edge and $m\neq \otimes$.  
	and  there must be an outer $\otimes$-edge $f=vw$.
	Let $x$ be a leaf such that $\lca(w,x)=u$. Since $\Ts$ is a
  phylogenetic tree, such a leaf always exists. Then $\eps(x,w)=m$
	if and only if $\lambda(e)=m$, i.e., the inner edge $e$ cannot
	be re-labeled. This establishes the final statement. 
\end{proof}

\section{Summary and Outlook}

We have considered maps $\eps: \irr{X \times X}\to \Mo$ and
edge labeled di-graphs $(K_{|X|},\eps)$
that can be explained by edge-labeled phylogenetic trees. Such graphs
generalize the notion of xenology and simple Fitch graphs \cite{Geiss:17,Hel:18}. 
As a main result, we gave a characterization of  Fitch graphs
based on  four simple conditions \emph{(T1)} to \emph{(T4)}
that are defined in terms of underlying edge-disjoint subgraphs. 
This in turn led to an $O(|X|^2)$-time  algorithm
to recognize Fitch graphs  $(K_{|X|},\eps)$ and for the reconstruction of 
the \emph{unique} least-resolved $\Mo$-edge-labeled phylogenetic tree that can explain them.

From the combinatorial point of view it might be of interest to consider more general maps
$\Eps:\irr{X\times X} \to \mathcal{P}(M) \cup \{\otimes\}$, where 
$\mathcal{P}(M)$ denotes the powerset of $M$.
In this case, there are a couple of ways to define
when $\Eps$ is tree-like. The two most obvious
ways, which we call ``Type-1''  and ``Type-2'' tree-like,
 are stated here.

\begin{description} 
	\item[\normalfont{\emph{The map $\Eps$ is tree-like }}]
	\item[\normalfont{\emph{of Type-1},}] if there is an edge-labeled tree $(T,\lambda)$ on $X$
such that for at least one $m\in \Eps(x,y)$ there is an edge 
on the path from  $\lca(x,y)$ to $y$ with label $m$.
	\item[\normalfont{\emph{of Type-2},}] if there is an edge-labeled tree $(T,\lambda)$ on $X$
such that for all $m\in \Eps(x,y)$ there is an edge 
on the path from  $\lca(x,y)$ to $y$ with label $m$.
\end{description}

Note, if $|M|=1$ or $|\Eps(x,y)|=1$ for all $x,y\in X$, 
then the problem of determining whether 
$\Eps$ is Type-1 or Type-2 tree-like
 reduces to the problem of determining whether $(K_{|X|}, \Eps)$
 is a Fitch graph or not. 
Moreover, if the sets $\Eps(x,y)$, $x,y\in X$ are pairwise disjoint,
we can define a set $N=\{m_{\Eps(x,y)} \mid x,y\in X\}$ of
symbols that identifies each symbol $m_{\Eps(x,y)}$ with the set $\Eps(x,y)$.
The established results imply the following 
\begin{corollary}
If the map $\eps:\irr{X\times X} \to N\cup \{\otimes\}$
with  $\eps(x,y)=m_{\Eps(x,y)}$ is tree-like, then
the map  $\Eps$ is  tree-like of Type-1. 
\end{corollary}

It would be of interest to understand  such generalized tree-like maps
in more detail. To this end, results established in \cite{SEMPLE1999300, BS:95, HSM:18} 
might shed some light on this question. 
Moreover, maps that cannot be explained by trees 
may be explained by phylogenetic networks, an issue that has not been addressed so-far.

\bibliographystyle{plain}
\bibliography{biblio}
\end{document}